\PassOptionsToPackage{unicode}{hyperref}
\PassOptionsToPackage{hyphens}{url}
\PassOptionsToPackage{dvipsnames,svgnames,x11names}{xcolor}

\documentclass[12pt]{article}
\usepackage{amsmath,amssymb,amsthm,mathtools,bm}
\usepackage{booktabs}
\usepackage{verbatim}
\usepackage{longtable}
\usepackage{makecell}
\usepackage{tabularx}
\usepackage{hyperref}
\usepackage{natbib}
\usepackage{graphicx}
\usepackage{subcaption}
\usepackage{float}
\usepackage{tikz}
\usepackage{pgfplots}
\pgfplotsset{compat=1.18}
\usetikzlibrary{arrows.meta,decorations.pathreplacing,positioning}
\usepackage{pdfpages}
\usepackage[ruled,vlined]{algorithm2e}

\hypersetup{
  pdftitle={A Leave-One-Out Influence Statistic for Density-Based Outlier Detection},
  pdfauthor={Aurélien Nicosia, Thierry Duchesne, Michel Carbon},
  pdfkeywords={outlier detection, influence diagnostics, leave-one-out, density estimation, binned estimators},
  colorlinks=true,
  linkcolor={blue},
  filecolor={Maroon},
  citecolor={Blue},
  urlcolor={Blue},
  pdfcreator={LaTeX}}

\theoremstyle{plain}

\newtheorem{proposition}{Proposition}

\theoremstyle{definition}

\newtheorem{assumption}{Assumption}

\newcommand{\R}{\mathbb{R}}

\newcommand{\blind}{0}

\begin{document}

\def\spacingset#1{\renewcommand{\baselinestretch}{#1}\small\normalsize}
\spacingset{1}

%%%%%%%%%%%%%%%%%%%%%%%%%%%%%%%%%%%%%%%%%%%%%%%%%%%%%%%%%%%%%%%%%%%%%%%%%%%%%%

\if0\blind
{
  \title{\bf A Leave-One-Out Influence Statistic for Density-Based Outlier Detection}
  \author{Aurélien Nicosia\\
    Université Laval\\
    \texttt{nicosia.aurelien@gmail.com}\\
    and \\
    Thierry Duchesne\\
    Université Laval\\
    and \\
    Michel Carbon\\
    Université Laval}
  \maketitle
} \fi

\if1\blind
{
  \bigskip
  \bigskip
  \bigskip
  \begin{center}
    {\LARGE\bf A Leave-One-Out Influence Statistic for Density-Based Outlier Detection}
\end{center}
  \medskip
} \fi
\date{}
\bigskip
\begin{abstract}
We propose a density-based leave-one-out influence score for unsupervised outlier detection. The motivation is that outliers are naturally associated with regions of very small probability density, but direct leave-one-out density refitting can be computationally prohibitive. We use the Linear-Blend Frequency Polygon (LBFP) estimator and define a score that compares the full-sample fitted density at an observation with the fitted density obtained after removing that observation, while keeping the grid and bandwidth fixed. The resulting statistic measures a relative density perturbation at the observation's own location. For the LBFP estimator, this score has an exact closed-form update, so the density estimator does not need to be refitted for each observation. This preserves a direct density interpretation while making the method computationally efficient for large samples. We study the score under contamination and show that regular positive-density observations and contamination-driven observations have distinct asymptotic orders. Simulations over a broad range of contamination models illustrate these theoretical regimes, show competitive performance relative to standard benchmarks, and document computing time. A credit-card fraud application with 29 variables illustrates that the method works well on a large real data set.
\end{abstract}

\noindent%
\vfill

\newpage
\spacingset{1.8} % JCGS spacing requirement
\graphicspath{{figures/}}

\section{Introduction}\label{sec:intro}

Outlier detection has long been a central problem in statistics and data science. Its goal is to identify observations that are unlikely under the mechanism generating the bulk of the data. This task is getting increasingly difficult in modern applications, where datasets may contain many observations, several variables and dependence between observations, which all create challenges for both statistical modeling and computation. These challenges are especially acute in outlier detection, because the regions of interest are precisely those where observations are rare.

Many unsupervised outlier detection methods have been proposed based on different notions of outlyingness. Density-based methods identify observations
that fall in regions of extremely low density, which is the most natural definition of an outlier. Other methods are based on properties of outlying observations, such as their distance to their nearest neighbors,
e.g., \(k\)NN \citep{Ramaswamy2000} or their isolation from data-dense regions, e.g., iForest \citep{liu2008isolationforest}. We believe that density-based methods have stronger theoretical foundations for outlier detection than distance-based methods, but in large samples or high-dimensional settings, their computational cost increases and their ability to detect outliers deteriorates substantially. Our goal in this paper is therefore to develop a density-based outlier identification method whose computational and inferential performance remains similar to that of distance or isolation-based approaches with large or
high-dimensional samples and that is robust to departures from independence among observations.

Density-based methods provide a natural framework for outlier detection because they score observations according to the amount of probability mass around them. In this view, an observation is unusual when it lies in a region where the fitted density is small, either in absolute terms or relative to nearby observations. Local density-contrast methods such as the Local Outlier Factor (LOF) compare the density around an observation with the density around its nearest neighbors \citep{Breunig2000}, while the Kernel Density Estimation Outlier Score (KDEOS) uses kernel-density estimates and evaluates observations relative to a neighborhood-based density baseline \citep{SchubertZimekKriegel2014}. A related strategy is to use a leave-one-out density score: if removing an observation substantially changes the fitted density at its own location, then that observation contributes unusually strongly to the fitted density at that point. Leave-one-out kernel-density methods use this principle for outlier scoring \citep{KandanaarachchiHyndman2021}.

Recently \citet{CarbonDuchesne2024LBFP,CarbonDuchesne2025} established the strong theoretical properties of the Linear-Blend Frequency Polygon (LBFP) multivariate estimator proposed by \citet{Scott1985}. In particular, they have shown its convergence under independent and dependent data, such as in strictly stationary alpha-mixing random fields. They have obtained closed-form expressions for the values of the smoothing parameters that yield good density estimation, even in low-density areas. However, using this estimator in a leave-one-out outlier detection approach would be computationally expensive. In this paper, we show that
the LBFP can indeed be used to define a leave-one-out density influence score. The proposed score compares the full-sample fitted density at an observation with the fitted density obtained after removing that observation, while keeping hyperparameters fixed. We develop a closed-form expression for the LBFP leave-one-out score so that the density estimator does not have to be refitted to compute the score for each observation. We also demonstrate its good theoretical properties under contamination and illustrate these properties through simulations over a broad range of contamination scenarios and application to a large real dataset on credit-card fraud. When compared with other density-based methods and with $k$NN and iForest, we conclude that the proposed method has a natural density-based interpretation with strong computational efficiency and a generally good ability to identify outlying observations.

The remainder of the paper is organized as follows. Section~\ref{sec:index} defines the leave-one-out density influence score, derives the closed-form expression for the leave-one-out LBFP update and studies the influence score's asymptotic behavior under contamination. Section~\ref{sec:simu} studies the computational and inferential performance of the proposed score under various contamination scenarios and sample sizes and compares the method to other benchmarks. The method is applied to credit card fraud detection in Section~\ref{sec:data}. Section~\ref{sec:discussion} discusses limitations and possible extensions.

\section{Leave-One-Out Density Influence Score}\label{sec:index}

\subsection{Definition and computation}\label{sec:index_definition}

Let $X_{{\bf i}}\in\mathbb{R}^d$, ${\bf i \in {\cal{I}_{\bf n}}}$, denote the $n=n_1\times\cdots\times n_N$ observations of a $d$-dimensional random vector on a rectangular lattice
${\cal{I}_{\bf n}}\subset \mathbb{Z}^N$. For example, we could have $N=2$, $i_1\in\{1,\dots,n_1\}$ and $i_2\in\{1,\dots,n_2\}$ if we observe
the random vector at $n_2$ regular time steps $t,2t,\dots,n_2t$ for each of $n_1$ individuals. 

Let $\widehat f_b$ denote
a fitted density estimate with $d$-dimensional smoothing parameter $b=(b_1,\ldots,b_d)$.
For each observation, let $\widehat f_{b,(-{\bf i})}$ denote the corresponding estimate computed
with the same value of $b$, but after deleting $X_{\bf i}$,
i.e., from the sample $X_{\bf j}$, ${\bf j}\in{\cal I}_{\bf n}\setminus\{{\bf i}\}$.
We define the leave-one-out density influence score (LOODIS) as
\begin{equation}
D_{\bf i}
=
1-
\frac{\widehat f_{b,(-{\bf i})}(X_{\bf i})}
{\widehat f_b(X_{\bf i})}
=
\frac{\widehat f_b(X_{\bf i})-\widehat f_{b,(-{\bf i})}(X_{\bf i})}
{\widehat f_b(X_{\bf i})}.
\label{eq:Di_def}
\end{equation}
Thus, $D_{\bf i}$ is a relative perturbation of the fitted density at the observation's own location, measuring the extent to which the fitted density at $X_{\bf i}$ depends on the presence of $X_{\bf i}$ itself. A large value of $D_{\bf i}$ therefore indicates strong local self-support under the fitted estimator rather than mere geometric isolation. If deleting $X_{\bf i}$ leaves the fitted density at $X_{\bf i}$ unchanged, then $D_{\bf i}=0$. The LOODIS is dimensionless because it is a ratio of density estimates. The term ``influence'' is used here in this finite-sample deletion sense, not in the infinitesimal influence-function sense of classical robustness theory \citep{Hampel1974}. 

The LOODIS defined in \eqref{eq:Di_def} can be used with different density estimators, but its usefulness depends on whether the leave-one-out density estimates can be computed efficiently and on the quality of the density estimator in data-sparse regions. For outlier detection purposes, where the mechanism that generates the outlying data points can be far from a standard well-behaved density, an estimator with very broad local convergence properties is desirable.
We therefore propose the use of the Linear-Blend Frequency Polygon (LBFP) estimator \citep{Scott1985}. 
The LBFP is a binned density estimator that linearly interpolates the multivariate histogram. Its convergence theory covers independent settings as well as weakly dependent settings, including strictly stationary alpha-mixing random fields \citep{CarbonDuchesne2024LBFP,CarbonDuchesne2025}.
Thanks to a computational property that we establish below, with the LBFP it is not necessary to compute the \(n\) leave-one-out estimators $\widehat f_{b,(-{\bf i})}$, ${\bf i}\in{\cal I}_{\bf n}$, separately; all of them can instead be obtained from a single computation of the estimator \(\widehat f_b\) based on the full sample.

To simplify notation, we take \(N=1\) in the remainder of the paper and write the observations as \(X_1,\ldots,X_n\). Proposition~\ref{prop:lbfp-closed-form} remains valid for \(N>1\) and general lattice indices: one simply replaces the scalar index \(i\) with the multi-index \({\bf i}\), because the leave-one-out update deletes a single observation while keeping the bin width and grid fixed. Consider $n$ observations of a $d$-dimensional random vector $X_1,\dots,X_n$. To compute the LBFP, we first define a grid of cells $I_k$ in $\R^d$ all of volume $V_b=b_1\times\cdots \times b_d$ 
with lower grid anchors
$x_k$, viz.
\[
I_k=\prod_{s=1}^d [x_{k,s},x_{k,s}+b_s).
\]
The grid limits are chosen coordinatewise with a half-bin margin around the sample, so that the smallest lower grid anchor satisfies \(\min_k x_{k,s}<\min_{1\le i\le n}X_{i,s}-b_s/2\) and the largest upper grid endpoint satisfies \(\max_k(x_{k,s}+b_s)>\max_{1\le i\le n}X_{i,s}+b_s/2\), for \(s=1,\ldots,d\).

Let $k(x)$ be the index of the cell containing $x$, and put
\[
u_s(x)=\frac{x_s-x_{k(x),s}}{b_s}\in[0,1),
\qquad s=1,\ldots,d.
\]
For each cell index $k$ and for $j=(j_1,\ldots,j_d)\in\{0,1\}^d$, we define $B_{k+j}$ as
\[
B_{k+j}
=\prod_{s=1}^d
\left[
x_{k,s}+\left(j_s-\frac12\right)b_s,\,
x_{k,s}+\left(j_s+\frac12\right)b_s
\right),
\]
which is cell $I_k$ shifted by half a cell width in each direction
corresponding to a nonzero entry in $j$.
Let $\nu_{k+j}$ denote the number of sample observations that fall in $B_{k+j}$.
Then the LBFP estimator is defined as
\begin{equation}
\widehat f_b(x)
=
\frac{1}{nV_b}
\sum_{j\in\{0,1\}^d}
c_j(x)\nu_{k(x)+j},
\label{eq:lbfp}
\end{equation}
where the weights $c_j(x)$ are given by
\[
c_j(x)
=
\prod_{s=1}^d
u_s(x)^{j_s}\{1-u_s(x)\}^{1-j_s}.
\]
We now establish the main computational result of this paper,
namely, a closed-form formula for the leave-one-out
LBFP calculation, which we prove in Appendix~\ref{app:closed_form}.
\begin{proposition}
\label{prop:lbfp-closed-form}
Let $u_{is}=u_s(X_i)$ and define $
W_i
=
\prod_{s=1}^d
\max\{u_{is},1-u_{is}\}.
$
Then
\begin{align*}
\widehat f_{b,(-i)}(X_i)
&=
\frac{n}{n-1}\widehat f_b(X_i)
-
\frac{W_i}{(n-1)V_b}\\
\Leftrightarrow D_i
&=
\frac{1}{n-1}
\left(
\frac{W_i}{V_b\widehat f_b(X_i)}
-1
\right).
\end{align*}
\end{proposition}

The weight $W_i$ in Proposition \ref{prop:lbfp-closed-form} is such that $2^{-d}\le W_i\le 1$. Writing $k_i=k(X_i)$, the weight $W_i$ can be viewed as a product of the relative positions of $X_i$ within the LBFP interpolation cell containing $X_i$ and its neighboring counting cells $B_{k_i+j}$, $j\in\{0,1\}^d$.
As a consequence, the LOODIS $D_i$ is bounded above by 1 and can take on slightly negative values
for observations in areas of very high density. In addition to the closed-form update in Proposition~\ref{prop:lbfp-closed-form}, the implementation uses a sparse representation of the LBFP counting grid, so that empty grid cells are neither stored nor scanned explicitly; this computational strategy is described in Supplementary Section~S3.

Now that we have demonstrated that the LBFP has the required computational property of
not requiring $n$ leave-one-out estimates, we must show that the LOODIS based on the
LBFP, which we henceforth refer to as LOO-LBFP, has the ability to identify outliers across several contamination scenarios.

\subsection{Ability of LOO-LBFP to identify outliers}\label{sec:discrimination}

We now study how the LOO-LBFP behaves under a classical contamination setting
\citep{Huber1964}. Suppose that $X_1,\dots,X_n$ are observations from a distribution with mixture density
\begin{equation}   
f_n(x)
=
(1-\varepsilon_n)f_0(x)
+
\varepsilon_n g_n(x),\label{eq:contaminmodele}
\end{equation}
where \(f_0\) denotes the inlier density, \(g_n\) denotes the contaminating (outlier) density,
and \(\varepsilon_n\) is the contamination proportion. The density estimator \(\widehat f_b\) 
computed from the contaminated sample now estimates the mixture density \(f_n\).
We let \(C_i\in\{0,1\}\) be the contamination
indicator, where \(C_i=0\) if $X_i$ is an inlier observation generated from
\(f_0\), and \(C_i=1\) if $X_i$ is an observation arising from the contamination density \(g_n\).

The next proposition establishes the asymptotic behavior of $D_i$ under the contamination model (\ref{eq:contaminmodele}). 
The key idea is that the magnitude of \(D_i\) depends on how much the fitted density at \(X_i\) is supported by neighboring observations 
from the contaminating component. Weak assumptions on the sampling of the observations and on
the contamination model (\ref{eq:contaminmodele}) are required.

\begin{assumption}
\label{ass:contamination-regions}
\itshape

There exists \(\delta>0\) such that
\[
P\{f_n(X_i)\ge \delta\mid C_i=0\}\to1
\]
as \(n\to\infty\).

For contaminating observations, the mixture density is contamination-dominated in the sense that, for every \(\eta>0\),
\[
P\Bigg(
\left|
\frac{f_n(X_i)}{\varepsilon_n g_n(X_i)}-1
\right|>\eta
\;\Bigg|\;
C_i=1
\Bigg)
\to0
\]
as \(n\to\infty\). In addition,
\[
V_b\varepsilon_n g_n(X_i)=O_p(1)
\]
conditionally on \(C_i=1\).
\end{assumption}

Assumption~\ref{ass:contamination-regions} separates the two regimes used in Proposition~\ref{prop:contam_orders}. For inlier observations, the mixture density at \(X_i\) is bounded away from zero with probability tending to one, which prevents the denominator in the closed-form expression for \(D_i\) from vanishing. For contaminating observations, the displayed probability statement is the definition of the conditional convergence in probability of \(f_n(X_i)/\{\varepsilon_n g_n(X_i)\}\) to one. It says that the mixture density at \(X_i\) is asymptotically governed by the contamination density. The additional bound \(V_b\varepsilon_n g_n(X_i)=O_p(1)\) controls the contamination scale relative to the LBFP cell volume.

\begin{assumption}
\label{ass:contamination-LBFP}
\itshape
For inlier observations, the fitted density is consistent for the mixture density in the sense that, for every \(\eta>0\),
\[
P\Bigg(
\bigg|
\frac{\widehat f_b(X_i)}{f_n(X_i)}-1
\bigg|>\eta
\;\Bigg|\;
C_i=0
\Bigg)
\to0
\]
as \(n\to\infty\).

For contaminating observations,
\[
P\{\widehat f_b(X_i)>0\mid C_i=1\}\to1
\]
as \(n\to\infty\), and
\[
\frac{\varepsilon_n g_n(X_i)}{\widehat f_b(X_i)}=O_p(1)
\]
conditionally on \(C_i=1\).

\end{assumption}

Assumption~\ref{ass:contamination-LBFP} gives sufficient conditions on the asymptotic behavior of $\widehat f_b(X_i)$. For inlier observations, the displayed probability statement is the definition of the conditional convergence in probability of \(\widehat f_b(X_i)/f_n(X_i)\) to one. For contaminating observations, we impose a weaker condition: the fitted density is positive with probability tending to one and is not asymptotically smaller than the contamination scale \(\varepsilon_n g_n(X_i)\). Together, Assumptions~\ref{ass:contamination-regions} and~\ref{ass:contamination-LBFP} give the density bounds needed to control the denominator in the closed-form expression for \(D_i\).

\begin{proposition}
\label{prop:contam_orders}
Assume that \(d\) is fixed and that \(V_b\to0\) and \(nV_b\to\infty\) as \(n\to\infty\). Under Assumptions~\ref{ass:contamination-regions} and~\ref{ass:contamination-LBFP},
\begin{align*}
    \text{if } C_i &= 0,\qquad D_i=O_p\{(nV_b)^{-1}\}, \\
    \text{if } C_i &= 1,\qquad D_i\, nV_b\, \varepsilon_n g_n(X_i)=O_p(1).
\end{align*}
\end{proposition}
In other words, regular inlier observations have order \(O_p\{(nV_b)^{-1}\}\). For contaminating observations, the product \(D_i nV_b\varepsilon_n g_n(X_i)\) remains bounded in probability. 
%In most applications, this means that the LOO-LBFP of inlier observations will have a distribution located close to zero and that outlier observations may have a distribution
%that will tend to be distributed above zero.
The proposition also explicitly shows that the impact of deleting an observation on $\widehat f_b$ decreases as the local contaminating support around this observation increases.
The proof of the proposition is given in Appendix~\ref{app:contam_proof}.

Hence, while the asymptotic behavior of $D_i$ is clear for inlier observations,
its asymptotic behavior for outlying observations depends on how
the smoothing parameters $b_1,\dots,b_d$, the contamination
proportion $\varepsilon_n$ and the contamination density $g_n$ behave
as $n\to\infty$.
To understand this behavior, 
consider the expected number of contaminating observations around $x$,
\[
\lambda_n(x)
=
n\varepsilon_n
P\{Z\in A_b(x)\},
\]
where 
$A_b(x)
=
\bigcup_{j\in\{0,1\}^d}
B_{k(x)+j}$
where $Z$ is a random variable with density $g_n$. From
the mean value theorem, $\lambda_n(x)
\approx
nV_b\varepsilon_n g_n(x)$, which is the local-count version of the scale appearing
for contaminating observations in Proposition \ref{prop:contam_orders}.

%First,  
%Suppose $C_i=0$ and $X_i=x$, then $\lambda_n(x)$ has three asymptotic behaviors. When this observation is isolated in $A_b(x)$, then the fitted density near $x$ depends %strongly on the observation itself, thus 
%if $\lambda_n(x)\to0$ as $n\to\infty$, then a contaminating observation has little mutual support from other contaminating observations in its LBFP neighborhood. 
%This is a favorable regime for separating such observations from regular inliers.
If \(\lambda_n(x)\to0\), then a contaminating observation has few other contaminating observations in its LBFP neighborhood. In this case, deleting this observation can remove a large fraction of the fitted density at its location and \(D_i\) can be large. This is the case where outlier identification should be easiest.

Second, if $\lambda_n(x)\to\lambda\in(0,\infty)$, $D_i=O_p(1)$. Because the distribution of $D_i$ goes to zero at rate $(nV_b)^{-1}$ for inlier observations, the distribution of $D_i$ for outlier observations should still be non-degenerate, and
the finite local-support regime can still produce separation between contaminating and inlier observations.

%Finally, if $\lambda_n(x)\to\infty$, $D_i$ will decrease to zero for contaminating observations, but at a slower rate than of inliers when the contamination density remains the main contribution to $f_n(X_i)$ 
Finally, if \(\lambda_n(x)\to\infty\), \(D_i\) also decreases for contaminating observations. However, this decrease is at a slower rate than the inlier order \(O_p\{(nV_b)^{-1}\}\) when \(\varepsilon_n g_n(X_i)=o(1)\) remains the dominant contribution to \(f_n(X_i)\) on the outlier region. These interpretations are illustrated through a simulation study in Section~\ref{sec:simu}.

\subsection{Choice of $b$}

The rate of convergence of $\widehat f_b$ to its target density depends on the choice of $b$, and so does the behavior of $D_i$.
The asymptotic calculations of \citet{CarbonDuchesne2025} motivate bin widths of order \(n^{-1/(d+4)}\), with coordinate-specific scale constants. These scale constants involve the sample standard deviations
$$
\hat{\sigma}_s=\sqrt{\frac{1}{n-1}\sum_{i=1}^n(X_{si}-\bar{X}_s)^2},\ s=1,\dots,d.$$
Because these $\hat{\sigma}_s$ can be sensitive to outlying observations and lead to overly large $b_s$, we use the robustified version
$$
\widehat\sigma^{\mathrm{rob}}_s
=
\min\left\{
\widehat\sigma_s,\,
\frac{\mathrm{IQR}_s}{1.349}
\right\},
$$
where $\mathrm{IQR}_s/1.349$ is the sample interquartile range of $X_s$ divided by the interquartile range of the standard normal distribution. In particular, the coordinate bin width is taken proportional to \(\widehat\sigma^{\mathrm{rob}}_s n^{-1/(d+4)}\). The resulting bin-width rule is used in all simulations and data applications. The corresponding method configurations and reproducibility scripts are documented in the Supplementary Materials, Sections~S1 and S2.

\section{Simulations} \label{sec:simu}

\subsection{Behavior of $D_i$}\label{sec:sim_lambda}

We run a first set of simulations to illustrate the distributional behavior of $D_i$ given by Proposition~\ref{prop:contam_orders} and the impact of the value
of the expected number of neighbors $\lambda$ on the distribution of $D_i$. 
We run a full factorial set of 24 simulations, with three values for sample size ($n\in\{1000,10000,100000\}$), two densities $f_0$ and four values
of $\lambda$ ($\lambda \in \{0.25, 1, 4, 16\}$). The contamination fraction is set to $\varepsilon_n=5\%$ and $d=2$ in all 24 simulations.

Two inlier densities $f_0$ are used. The first is the standard centered isotropic bivariate normal,
$N_2\{(0,0)^\top,I_2\}$, where $I_2$ is the $2\times2$ identity matrix. The second is the four-state Markov-switching Gaussian mixture from the simulation study in \citet{CarbonDuchesne2025}. 

In all simulations, outliers are drawn from a bivariate normal distribution
$N_2\{(6,6)^\top,\tau_n^2 I_2\}$,
with $\tau_n$ set so as to obtain the target $\lambda$. 
The details of this calibration and the realized empirical values of \(\widehat\lambda\) are reported in the Supplementary Materials, Section~S1.

Figure~\ref{fig:lambda_regimes} displays 24 pairs of boxplots of the $D_i$ values, each pair consisting of one boxplot for inlier observations (grey) and one
boxplot for contaminating observations (red). The simulations show how the distribution of \(D_i\) changes with the expected number of nearby contaminating observations. The \(D_i\) values for inliers decrease toward the order \(O_p\{(nV_b)^{-1}\}\) as \(n\) increases. The values of \(D_i\) for contaminating observations are close to one when \(\lambda\) is small and move away from one as \(\lambda\) increases.

\begin{figure}[H]
    \centering
    \includegraphics[width=0.95\textwidth]{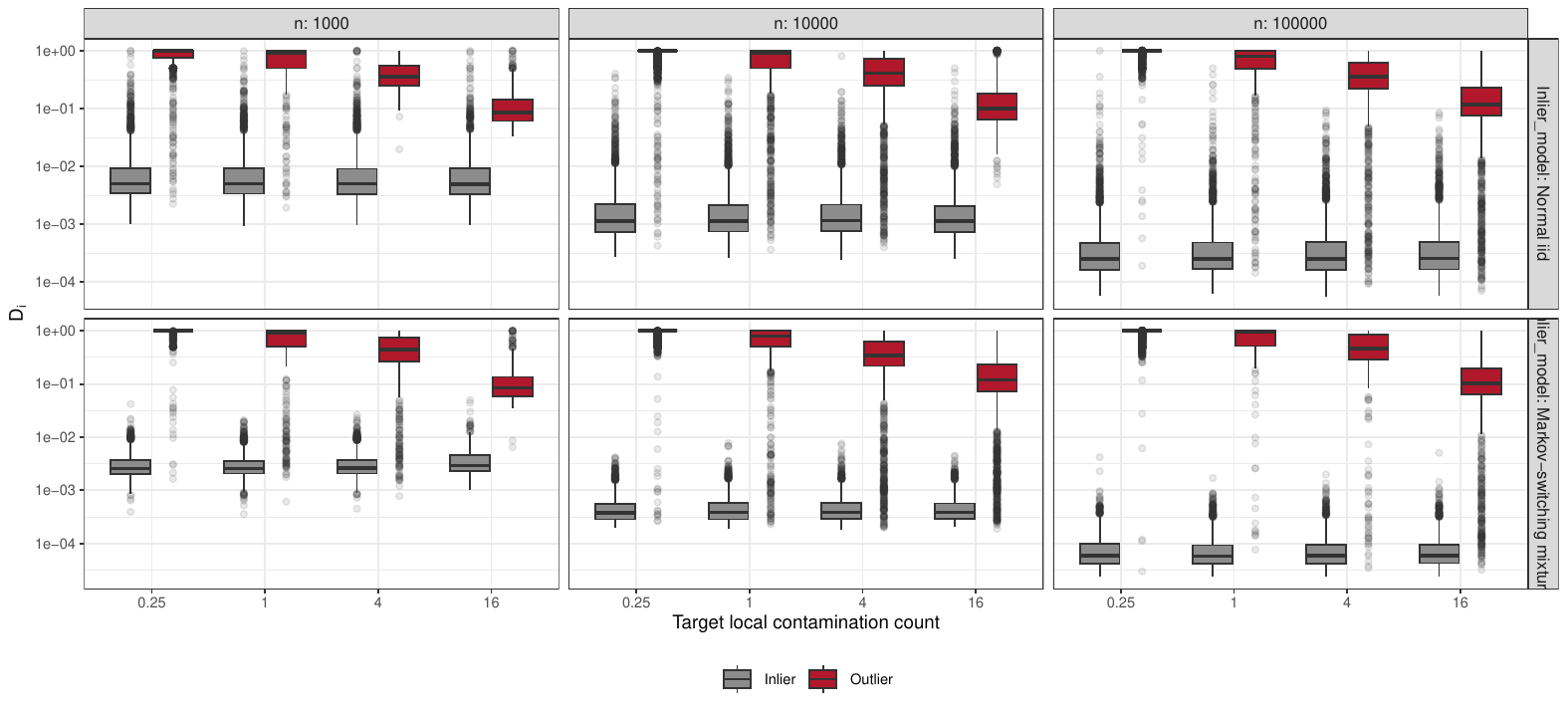}
    \caption{Distribution of $D_i$ by class, sample size, inlier model, and $\lambda$. Rows correspond to data-generating models and columns to sample sizes. Small $\lambda$ corresponds to weak local support among outliers; large $\lambda$ corresponds to stronger mutual support. The y-axis is on a logarithmic scale.}
    \label{fig:lambda_regimes}
\end{figure}

\subsection{Robustness to contamination}\label{sec:sim_comparison}

We next compare LOO-LBFP with standard baselines under four contamination scenarios. The goal is not to show dominance of the LOO-LBFP, but to assess how this score behaves across distinct anomaly mechanisms. We consider two types of outlier detection methods, namely density-based outlier detection criteria and non-density-based
criteria. The density-based criteria are LOO-KDE \citep{KandanaarachchiHyndman2021}, LOF \citep{Breunig2000}, and KDEOS \citep{SchubertZimekKriegel2014}; the non-density-based criteria are $k$NN distance \citep{Ramaswamy2000} and iForest \citep{liu2008isolationforest}.
We simulated 100 samples from each of the four models. For each simulated sample, $n=1000$ and the contamination fraction $\varepsilon_n=5\%$.
\begin{figure}[H]
    \centering
    \includegraphics[width=0.88\textwidth]{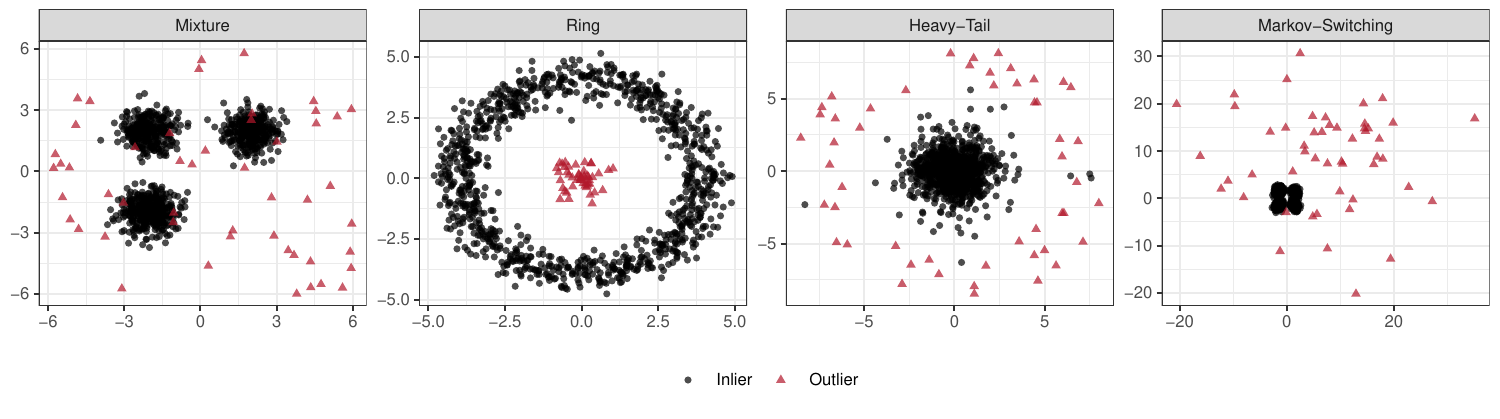}
\caption{Representative realizations of the four simulated scenarios: (a) Mixture, three Gaussian clusters with uniform global contamination; (b) Ring, noisy ring with central Gaussian contamination; (c) Heavy-Tail, Gaussian and Student-$t_3$ inliers with radial extremes; and (d) Markov-Switching, a dependent four-state Gaussian mixture with Gaussian contamination \(N_2\{(6,6)^\top,10^2I_2\}\). Black circles are inliers and red triangles are outliers.}
    \label{fig:sim_scenarios}
\end{figure}

Figure~\ref{fig:sim_scenarios} shows a representative realization from each of the four simulation models. In all scenarios, the data are generated from the contamination model
\[
f_n(x)=0.95f_0(x)+0.05 g_n(x),
\]
where \(f_0\) is the inlier distribution and \(g_n\) is the contaminating distribution. The four choices of \(f_0\) and \(g_n\) were selected to represent distinct anomaly mechanisms.
The details of the data-generating mechanisms, software implementations, seeds, and hyperparameters for all methods are documented in the Supplementary Materials, Section~S1.

In the Mixture scenario, \(f_0\) is a three-component Gaussian mixture with centers \((-2,-2)^\top\), \((2,2)^\top\), and \((-2,2)^\top\), corresponding mixture weights \(0.4\), \(0.3\), and \(0.3\), and common covariance matrix \(0.3I_2\). The contaminating density \(g_n\) is uniform on \([-6,6]^2\). This scenario represents diffuse global contamination around a multimodal inlier distribution.

In the Ring scenario, \(f_0\) is the distribution of $
X=(R\cos\Theta,R\sin\Theta)^\top$, where $ \Theta\sim U(0,2\pi)$, $
R\sim N(4,0.4^2)$
and \(g_n=N_2\big((0,0)^\top,0.2I_2\big)\). This scenario should be less favorable to scores purely based on observation deletion.

In the Heavy-Tail scenario, $f_0$ is a $70/30$ mixture of a bivariate normal distribution $N_2\big((0,0)^\top,I_2\big)$ and a bivariate Student-$t$ distribution with 3 degrees of freedom and independent components. For the contaminating density, $
X=(R\cos\Theta,R\sin\Theta)^\top,
$ where $
\Theta\sim U(0,2\pi)$ and $
R\sim U(6,9)$.
This setting tests whether the methods can identify low-support extremes when the inlier distribution itself has heavier tails than a Gaussian model.

In the Markov-Switching scenario, \(f_0\) is a dependent four-state Markov-switching Gaussian mixture from \citet{CarbonDuchesne2024LBFP}, with \(g_n=N_2\{(6,6)^\top,10^2I_2\}\). This scenario compares the methods when the data are not independent.

For each scenario, we simulate $100$ samples of size $n=1000$ and we compute the outlier score for the six methods for every observation. To illustrate Proposition~\ref{prop:contam_orders}, we compare the across-replication means and standard deviations of the raw score $D_i$ separately for inlier and outlier observations, for each scenario. Table~\ref{tab:D_summary} reports these results. As expected, the LOO-LBFP values are well separated between outlier and inlier observations. This is true for every scenario, albeit less pronounced in the Ring scenario, where the outliers are more densely concentrated.

\begin{table}[H]
    \centering
    \caption{LOO-LBFP score summaries by scenario. The table reports across-replication means and standard deviations of the within-class average LOO-LBFP score $D_i$.}
    \begin{tabular}{lcc}
\toprule
Scenario & $D_{i}(C_i=0)$ & $D_i(C_i=1)$ \\
\midrule
Mixture & 0.0039 (0.0002) & 0.1342 (0.0220) \\
Ring & 0.0057 (0.0001) & 0.0076 (0.0007) \\
Heavy-Tail & 0.0251 (0.0029) & 0.6406 (0.0471) \\
Markov-Switching & 0.0032 (0.0001) & 0.7294 (0.0417) \\
\bottomrule
\end{tabular}
\label{tab:D_summary}

\end{table}

For each sample and method, we computed the area under the receiver operating characteristic curve (ROC-AUC) with outlier indicator as the target and outlier score as the predictor. Table~\ref{tab:sim_auc} reports the mean (standard deviation) of the 100 ROC-AUC values for each scenario and for each method. Values of ROC-AUC close to 1 correspond to near-perfect discrimination between inlier and outlier observations. Although LOO-LBFP is never the best-performing method, it consistently performs close to the best. Therefore, LOO-LBFP is a robust criterion for outlier detection. We will now compare its computational cost to that of the other methods.

\begin{table}[H] 
    \centering
    \caption{ROC-AUC (Mean (SD) over 100 replications). We compare the proposed LOO-LBFP with the density-based methods LOO-KDE, LOF, and KDEOS, and with the non-density-based methods $k$NN and iForest, across four scenarios.}
    \begin{tabular}{lcccc}
\toprule
Method & Mixture & Ring & Heavy-Tail & Markov-Switching \\
\midrule
LOO-LBFP & 0.904 (0.035) & 0.718 (0.066) & 0.990 (0.003) & 0.986 (0.012) \\
LOO-KDE & 0.960 (0.020) & 0.862 (0.043) & 0.988 (0.003) & 0.996 (0.006) \\
LOF & 0.956 (0.021) & 0.558 (0.030) & 0.973 (0.016) & 0.952 (0.036) \\
KDEOS & 0.942 (0.026) & 0.573 (0.052) & 0.945 (0.045) & 0.944 (0.050) \\
\hline
kNN & 0.960 (0.019) & 0.744 (0.093) & 0.997 (0.002) & 0.996 (0.006) \\
iForest & 0.957 (0.021) & 0.655 (0.079) & 0.997 (0.002) & 0.995 (0.006) \\
\bottomrule
\end{tabular}
\label{tab:sim_auc}

\end{table}

\subsection{Computational cost}\label{sec:sim_cost}

We computed the execution times of each method for every simulated sample. We report the mean (standard deviation) of the 100 execution times (in seconds) in Table~\ref{tab:sim_runtime}. We can see that LOO-LBFP is always among the fastest methods. This stems from Proposition~\ref{prop:lbfp-closed-form}, which eliminates the need for repeated leave-one-out refitting.

\begin{table}[H]
    \centering
    \caption{Execution times in seconds (Mean (SD) over 100 replications). We compare the proposed LOO-LBFP against LOO-KDE, LOF, KDEOS, $k$NN, and iForest across four scenarios.}
    \begin{tabular}{lcccc}
\toprule
Method & Mixture & Ring & Heavy-Tail & Markov-Switching \\
\midrule
LOO-LBFP & 0.007 (0.001) & 0.006 (0.001) & 0.007 (0.001) & 0.007 (0.001) \\
LOO-KDE & 1.229 (0.272) & 1.167 (0.084) & 1.173 (0.159) & 1.149 (0.092) \\
LOF & 0.005 (0.003) & 0.005 (0.001) & 0.005 (0.001) & 0.005 (0.000) \\
KDEOS & 0.114 (0.026) & 0.110 (0.004) & 0.114 (0.013) & 0.111 (0.008) \\
\hline
kNN & 0.003 (0.001) & 0.003 (0.000) & 0.003 (0.000) & 0.003 (0.001) \\
iForest & 0.010 (0.005) & 0.010 (0.000) & 0.009 (0.002) & 0.009 (0.005) \\
\bottomrule
\end{tabular}
\label{tab:sim_runtime}

\end{table}

To see how the computational cost increases with $n$, we simulated samples of varying sizes ($n\in \{500, 1000, 2000, 5000, 10000, 50000, 100000\}$) from the $N_2\big((0,0)^\top,I_2\big)$ distribution. No contamination is added in this experiment because the goal is to see how the computational time scales with sample size.

We simulated one sample for each value of $n$, as the computational time did not vary significantly from sample to sample. Figure~\ref{fig:runtime_by_n} summarizes the results. LOO-KDE is by far the method that is most sensitive to the sample size; we had to stop its execution after 5 minutes for $n > 10000$. As the sample size increases, LOO-LBFP seems to be the most computationally efficient method among those considered.
%The reasons underlying this behavior are discussed in Appendix~\ref{app:sim-computation}.

\begin{figure}[H]
    \centering
    \includegraphics[width=0.72\textwidth]{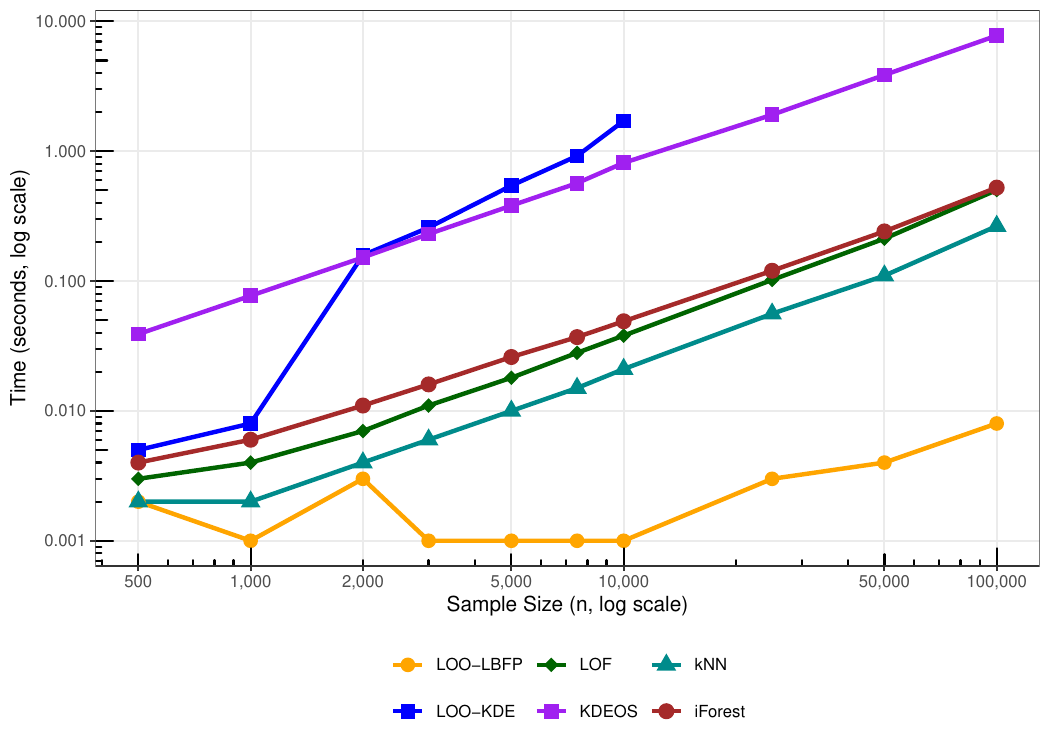}
    \caption{Execution time as a function of sample size. Curves are method-specific, and both axes are on a logarithmic scale.}
    \label{fig:runtime_by_n}
\end{figure}

\section{Credit Card Fraud Application}\label{sec:data}

We now consider the Credit Card Fraud Detection dataset \citep{dalpozzolo2015calibrating}, publicly available through OpenML (dataset ID 42175; \citealp{Vanschoren2013OpenML,creditcard2013}). The dataset contains $n=284{,}807$ transactions, of which 492 are fraudulent (for a fraud rate of $0.172\%$). Available variables are the 28 anonymized PCA variables V1--V28, Time, and Amount. Outlier detection is frequently
used to identify potential fraud \citep{BoltonHand2002,Chandola2009}. We therefore apply the methods compared in Section~\ref{sec:simu} to these data to see if the transactions with high outlier scores tend to correspond to fraudulent transactions. The scores are computed from the 28 PCA variables V1--V28 and Amount, all standardized before fitting; the variable Time is excluded because it records elapsed time within the data collection window rather than a transaction-level feature. It will be interesting to see if these methods work in this 29-dimensional setting.

For each method, we compute ROC-AUC with Fraud as the outcome and the outlier score as the predictor, using the full sample. We also compute the lift at the empirical fraud rate: it is the number of actual fraud cases among the 492 transactions with the highest scores divided by the fraud rate $492/284{,}807$. The higher the lift, the better the method is, and random sorting of the transactions would yield an expected lift of 1. We also record the computing time for each method. We do not have ROC-AUC or lift for LOO-KDE because it still had not finished its execution after 30 minutes.

\begin{table}[H]
    \centering
    \caption{Performance of five methods to detect fraudulent transactions on the full sample credit card dataset. Lift@0.172\% is precision among the top 492 ranked transactions divided by the overall fraud rate.}
    \label{tab:fraud_full_methods}
    \begin{tabular}{@{}lccc@{}}
\toprule
Method & ROC-AUC & Lift@0.172\% & Time (s) \\
\midrule
LOO-LBFP & 0.901 & 21.18 & 91.43 \\
LOO-KDE & NA & NA & $>1800$\\
LOF & 0.510 & 0.00 & 1246.27 \\
KDEOS & 0.879 & 0.00 & 615.17 \\
\hline
kNN & 0.956 & 65.89 & 585.00 \\
iForest & 0.951 & 134.13 & 1.50 \\
\bottomrule
\end{tabular}

\end{table}

Even though the methods are not trained to detect fraud, several of them assign high ranks to fraudulent transactions. The $k$NN and iForest scores achieve the largest lift values. LOO-LBFP has lower, but still quite respectable, lift and is substantially faster than $k$NN in this full-sample run. In this example, LOO-LBFP is the only density-based
criterion considered that is both computationally efficient and able to identify outlying observations, with a detection performance that does not lag too much behind non-density-based approaches.

\section{Conclusion}\label{sec:discussion}

Our objective in this paper was to propose a versatile and robust density-based outlier detection method that performs well in a vast array of contamination scenarios. We also wanted the
method to be computationally efficient so that it can be applied on large samples with several variables in reasonable time. We showed that the LOO-LBFP possesses all of these properties.
We first obtained a closed-form expression to compute this score without the need for repeated leave-one-out refitting. We then formally proved that the distribution of the LOO-LBFP
score was well separated between outlier and inlier observations under a general contamination model.

The results of our numerical investigations aligned with our theoretical findings. The distributions of the LOO-LBFP scores were well separated for inlier and outlier observations
in many contamination scenarios. When compared to other outlier detection methods, LOO-LBFP performed generally well with good computational speed. These conclusions also held true
when the methods were applied to a real dataset on credit card fraud with 29 variables and 284{,}807 observations.

Yet, the LOO-LBFP can probably still be improved. Because it so heavily relies on the quality of the density estimation around a point $X_i$, using
a smoothing parameter $b$ that varies with $x$ would likely lead to improved performance. While the computational cost of LOO-LBFP grows relatively slowly in $n$, it is
highly sensitive to the dimension $d$, as the LBFP is a weighted average of $2^d$ cell counts. If the data are sparse, as is often the case when one wants to identify outliers and was the case in the simulations and the credit card
fraud application, LOO-LBFP can be computed quickly. But with dense data in high dimensions, the method would break down. Other density-based methods do not
break down computationally when $d$ is large, but they perform poorly statistically. Another open question is how to use LOO-LBFP to formally assign an
inlier/outlier label to each observation. If the proportion of outliers, say $p$, in the sample is known, then labeling the $pn$ observations with the highest
LOO-LBFP score as outliers would be the natural approach, but developing a strategy to set this threshold when $p$ is unknown would be valuable.

\section*{Disclosure Statement}\label{disclosure-statement}

The authors have no conflicts of interest to declare.

\appendix
\section{Proof of Proposition~\ref{prop:lbfp-closed-form}}\label{app:closed_form}

\begin{proof}[Proof of Proposition~\ref{prop:lbfp-closed-form}]
Throughout the proof, the bin-width vector $b$, the grid anchor $x$, and the counting cells are held fixed. Thus the leave-one-out operation removes only the contribution of $X_i$; it does not reselect the bandwidth or rebuild the grid.

We first derive an algebraic identity for any estimator that can be written as a normalized sum of individual contributions,
\[
\widehat f_b(x)
=
\frac{1}{nV_b}
\sum_{r=1}^n w_r(x),
\]
where $w_r(x)\ge0$. After deleting $X_i$,
\[
\widehat f_{b,(-i)}(x)
=
\frac{1}{(n-1)V_b}
\sum_{r\ne i} w_r(x).
\]
At the evaluation point $x=X_i$, define
\[
W_i=w_i(X_i),
\qquad
S_i=\sum_{r\ne i}w_r(X_i).
\]
Then
\[
\widehat f_b(X_i)
=
\frac{W_i+S_i}{nV_b},
\qquad
\widehat f_{b,(-i)}(X_i)
=
\frac{S_i}{(n-1)V_b}.
\]
Solving the first identity for $S_i$ gives
\[
S_i=nV_b\widehat f_b(X_i)-W_i.
\]
Substitution into the leave-one-out expression yields
\[
\widehat f_{b,(-i)}(X_i)
=
\frac{n}{n-1}\widehat f_b(X_i)
-
\frac{W_i}{(n-1)V_b}.
\]
Therefore
\[
D_i
=
1-
\frac{\widehat f_{b,(-i)}(X_i)}
{\widehat f_b(X_i)}
=
\frac{1}{n-1}
\left\{
\frac{W_i}{V_b\widehat f_b(X_i)}
-1
\right\}.
\]
Because $W_i\ge0$ and $S_i\ge0$, the same decomposition implies
\[
0
\le
\widehat f_{b,(-i)}(X_i)
\le
\frac{n}{n-1}\widehat f_b(X_i).
\]
Hence, whenever $\widehat f_b(X_i)>0$, the fitted densities are nonnegative and
\[
-\frac{1}{n-1}\le D_i\le 1.
\]

It remains to identify $W_i$ for the LBFP estimator. Let $X_i$ lie in the interpolation cell $I_k$. The LBFP fitted value is
\[
\widehat f_b(X_i)
=
\frac{1}{nV_b}
\sum_{j\in\{0,1\}^d}
c_j(X_i)\nu_{k+j}.
\]
Writing the counts as sums of indicators,
\[
\nu_{k+j}
=
\sum_{r=1}^n
\mathbf 1\{X_r\in B_{k+j}\},
\]
shows that the contribution of $X_i$ to its own fitted value is the interpolation weight attached to the unique neighboring counting cell that receives $X_i$. Denote this index by $j_i^\star\in\{0,1\}^d$. Coordinatewise,
\[
j_{is}^\star
=
\begin{cases}
0, & u_{is}<1/2,\\
1, & u_{is}\ge 1/2,
\end{cases}
\]
with ties being immaterial because the two coordinate weights are equal at $u_{is}=1/2$. Thus
\[
W_i=c_{j_i^\star}(X_i).
\]
Using the multilinear weight definition,
\[
c_{j_i^\star}(X_i)
=
\prod_{s=1}^d
u_{is}^{j_{is}^\star}
(1-u_{is})^{1-j_{is}^\star}
=
\prod_{s=1}^d
\max\{u_{is},1-u_{is}\}.
\]
Each coordinate factor lies in $[1/2,1]$, so $2^{-d}\le W_i\le1$. Substituting this expression for $W_i$ into the algebraic identity above gives the stated closed form.
\end{proof}

\section{Proof of Proposition~\ref{prop:contam_orders}}\label{app:contam_proof}

\begin{proof}[Proof of Proposition~\ref{prop:contam_orders}]
By Proposition~\ref{prop:lbfp-closed-form},
\[
D_i
=
\frac{W_i}{(n-1)V_b\widehat f_b(X_i)}
-\frac{1}{n-1}.
\]
For the LBFP estimator, \(2^{-d}\le W_i\le1\). Since \(d\) is fixed, \(W_i=O(1)\) deterministically.

We first consider inlier observations, that is, the conditional distribution given \(C_i=0\). By Assumption~\ref{ass:contamination-regions}, there exists \(\delta>0\) such that
\[
P\{f_n(X_i)\ge\delta\mid C_i=0\}\to1.
\]
By Assumption~\ref{ass:contamination-LBFP}, for every \(\eta>0\),
\[
P\left(
\left|
\frac{\widehat f_b(X_i)}{f_n(X_i)}-1
\right|>\eta
\;\middle|\;
C_i=0
\right)
\to0.
\]
Taking \(\eta=1/2\), the event
\[
\left\{
f_n(X_i)\ge\delta
\right\}
\cap
\left\{
\left|
\frac{\widehat f_b(X_i)}{f_n(X_i)}-1
\right|\le\frac12
\right\}
\]
has conditional probability tending to one given \(C_i=0\). On this event,
\(\widehat f_b(X_i)\ge\delta/2\). Hence
\[
\frac{1}{\widehat f_b(X_i)}=O_p(1)
\]
conditionally on \(C_i=0\). Therefore
\[
\frac{W_i}{(n-1)V_b\widehat f_b(X_i)}
=
O_p\{(nV_b)^{-1}\}.
\]
The second term in the closed-form expression is \(O(n^{-1})\), which is absorbed in
\(O_p\{(nV_b)^{-1}\}\) because \(V_b\to0\). Hence
\[
D_i
=
O_p\{(nV_b)^{-1}\}
\]
conditionally on \(C_i=0\).

We now consider contaminating observations, that is, the conditional distribution given \(C_i=1\). By Assumption~\ref{ass:contamination-LBFP},
\[
P\{\widehat f_b(X_i)>0\mid C_i=1\}\to1.
\]
Thus the following calculation is valid on an event whose conditional probability tends to one. Multiplying the closed-form expression by
\(nV_b\varepsilon_n g_n(X_i)\), we obtain
\[
D_i\,nV_b\varepsilon_n g_n(X_i)
=
\frac{n}{n-1}
W_i
\frac{\varepsilon_n g_n(X_i)}{\widehat f_b(X_i)}
-
\frac{n}{n-1}
V_b\varepsilon_n g_n(X_i).
\]
The first term on the right-hand side is \(O_p(1)\) conditionally on \(C_i=1\) by Assumption~\ref{ass:contamination-LBFP}, because \(W_i=O(1)\) and
\[
\frac{\varepsilon_n g_n(X_i)}{\widehat f_b(X_i)}=O_p(1)
\]
conditionally on \(C_i=1\). The second term is \(O_p(1)\) conditionally on \(C_i=1\) by Assumption~\ref{ass:contamination-regions}, because
\[
V_b\varepsilon_n g_n(X_i)=O_p(1)
\]
conditionally on \(C_i=1\). Therefore
\[
D_i\,nV_b\varepsilon_n g_n(X_i)
=
O_p(1)
\]
conditionally on \(C_i=1\). This proves the result.
\end{proof}

\section*{Supplementary Materials}\label{sec:supplementary}

The supplementary materials document the simulation designs, method configurations, lambda-regime calibration table, credit-card preprocessing details, computing environment, and reproducibility scripts. The accompanying replication package is organized into \texttt{scripts/}, \texttt{data/}, \texttt{figures/}, \texttt{tables/}, and \texttt{article/} directories; its \texttt{README.md} lists the quick validation and full reproduction commands, and \texttt{validate\_submission.R} checks the R package, required artifacts, and exclusion of the raw credit-card data file. This replication package will be made available in a public archival repository upon publication.

\section*{Data Availability Statement}
The Credit Card Fraud Detection data used in this article are publicly available through OpenML, dataset ID 42175 \citep{Vanschoren2013OpenML,creditcard2013}. The replication code can retrieve the data from OpenML or use a local, non-versioned copy named \texttt{data/creditcard\_42175.rds}.

\section*{Code Availability Statement}
The code supporting this article is included in the accompanying replication package for review. A public archival version will be released after acceptance or publication.

\bibliography{article2_references}

@article{Scott1985,
  author  = {Scott, D. W.},
  title   = {Frequency polygons, theory and applications},
  journal = {Journal of the American Statistical Association},
  year    = {1985},
  volume  = {80},
  pages   = {348--354}
}

@article{Hampel1974,
  author  = {Hampel, Frank R.},
  title   = {The Influence Curve and Its Role in Robust Estimation},
  journal = {Journal of the American Statistical Association},
  year    = {1974},
  volume  = {69},
  number  = {346},
  pages   = {383--393},
  doi     = {10.1080/01621459.1974.10482962},
  url     = {https://doi.org/10.1080/01621459.1974.10482962}
}

@inproceedings{liu2008isolationforest,
  title        = {Isolation Forest},
  author       = {Liu, Fei Tony and Ting, Kai Ming and Zhou, Zhi-Hua},
  booktitle    = {2008 Eighth IEEE International Conference on Data Mining},
  year         = {2008},
  pages        = {413--422},
  publisher    = {IEEE},
  doi          = {10.1109/ICDM.2008.17},
  url          = {https://doi.org/10.1109/ICDM.2008.17}
}

@inproceedings{SchubertZimekKriegel2014,
  author    = {Erich Schubert and Arthur Zimek and Hans{-}Peter Kriegel},
  title     = {Generalized Outlier Detection with Flexible Kernel Density Estimates (KDEOS)},
  booktitle = {Proceedings of the 2014 {SIAM} International Conference on Data Mining (SDM)},
  year      = {2014},
  pages     = {542--550},
  doi       = {10.1137/1.9781611973440.63},
  url       = {https://epubs.siam.org/doi/10.1137/1.9781611973440.63}
}

@article{CarbonDuchesne2025,
  author  = {Michel Carbon and Thierry Duchesne},
  title   = {Asymptotic normality of multivariate frequency polygons for stationary random fields},
  journal = {Annals of the Institute of Statistical Mathematics},
  year    = {2025},
  doi     = {10.1007/s10463-025-00952-x},
  url     = {https://link.springer.com/article/10.1007/s10463-025-00952-x}
}

@article{CarbonDuchesne2024LBFP,
  author  = {Michel Carbon and Thierry Duchesne},
  title   = {Multivariate frequency polygon for stationary random fields},
  journal = {Annals of the Institute of Statistical Mathematics},
  year    = {2024},
  volume  = {76},
  number  = {2},
  pages   = {263--287},
  doi     = {10.1007/s10463-023-00883-5},
  url     = {https://doi.org/10.1007/s10463-023-00883-5}
}

@manual{GLBFPpackage,
  title  = {GLBFP: General Linear Blend Frequency Polygon Density Estimation},
  author = {Aurélien Nicosia and Thierry Duchesne and Michel Carbon},
  year   = {2026},
  note   = {R package version 0.5.2},
  url    = {https://github.com/AurelienNicosiaULaval/GLBFP},
  doi    = {10.5281/zenodo.17945962}
}

@misc{creditcard2013,
  title        = {Credit Card Fraud Detection},
  author       = {{European cardholders dataset}},
  year         = {2013},
  note         = {Available on Kaggle at \url{https://www.kaggle.com/mlg-ulb/creditcardfraud},
                  mirrored on OpenML (IDs 1597 and 42175)},
  howpublished = {\url{https://www.openml.org/d/42175}}
}

@article{Vanschoren2013OpenML,
  author  = {Vanschoren, Joaquin and van Rijn, Jan N. and Bischl, Bernd and Torgo, Luis},
  title   = {{OpenML}: Networked Science in Machine Learning},
  journal = {{SIGKDD} Explorations},
  year    = {2013},
  volume  = {15},
  number  = {2},
  pages   = {49--60},
  doi     = {10.1145/2641190.2641198},
  url     = {https://doi.org/10.1145/2641190.2641198}
}

@article{Breunig2000,
  author    = {Markus M. Breunig and Hans{-}Peter Kriegel and Raymond T. Ng and Jörg Sander},
  title     = {LOF: Identifying Density-Based Local Outliers},
  journal   = {ACM SIGMOD Record},
  volume    = {29},
  number    = {2},
  pages     = {93--104},
  year      = {2000},
  doi       = {10.1145/335191.335388}
}

@article{KandanaarachchiHyndman2021,
  author = {Kandanaarachchi, Sevvandi and Hyndman, Rob J.},
  title = {Leave-One-Out Kernel Density Estimates for Outlier Detection},
  journal = {Journal of Computational and Graphical Statistics},
  volume = {31},
  number = {2},
  pages = {586--599},
  year = {2021},
  doi = {10.1080/10618600.2021.2000425}
}

@book{Silverman1986,
  author    = {Silverman, B. W.},
  title     = {Density Estimation for Statistics and Data Analysis},
  publisher = {Chapman and Hall},
  address   = {London},
  year      = {1986}
}

@article{BoltonHand2002,
  author  = {Bolton, Richard J. and Hand, David J.},
  title   = {Statistical Fraud Detection: A Review},
  journal = {Statistical Science},
  volume  = {17},
  number  = {3},
  pages   = {235--255},
  year    = {2002},
  doi     = {10.1214/ss/1042727940}
}

@article{Chandola2009,
  title={Anomaly detection: A survey},
  author={Chandola, Varun and Banerjee, Arindam and Kumar, Vipin},
  journal={ACM computing surveys (CSUR)},
  volume={41},
  number={3},
  pages={1--58},
  year={2009},
  publisher={ACM New York, NY, USA},
  doi={10.1145/1541880.1541882}
}

@article{dalpozzolo2015calibrating,
  author = {Dal Pozzolo, Andrea and Caelen, Olivier and Johnson, Reid A and Bontempi, Gianluca},
  title = {Calibrating Probability with Undersampling for Unbalanced Classification},
  journal = {2015 IEEE Symposium Series on Computational Intelligence},
  year = {2015},
  pages = {159--166},
  doi = {10.1109/SSCI.2015.33}
}

@inproceedings{Ramaswamy2000,
  author    = {Ramaswamy, Sridhar and Rastogi, Rajeev and Shim, Kyuseok},
  title     = {Efficient Algorithms for Mining Outliers from Large Data Sets},
  booktitle = {Proceedings of the 2000 ACM SIGMOD International Conference on Management of Data},
  year      = {2000},
  pages     = {427--438},
  publisher = {ACM},
  address   = {New York, NY, USA},
  doi       = {10.1145/342009.335437}
}

@article{Huber1964,
  author  = {Huber, Peter J.},
  title   = {Robust Estimation of a Location Parameter},
  journal = {The Annals of Mathematical Statistics},
  volume  = {35},
  number  = {1},
  pages   = {73--101},
  year    = {1964},
  doi     = {10.1214/aoms/1177703732}
}

@book{Scott1992,
  author    = {Scott, D. W.},
  title     = {Multivariate Density Estimation: Theory, Practice, and Visualization},
  publisher = {John Wiley \& Sons},
  address   = {New York},
  year      = {1992}
}

\end{document}

\title{Supplementary Materials for\\A Leave-One-Out Influence Statistic for Density-Based Outlier Detection}
\author{Aurélien Nicosia \and Thierry Duchesne \and Michel Carbon}
\date{}
\maketitle

\def\spacingset#1{\renewcommand{\baselinestretch}{#1}\small\normalsize}
\spacingset{1.3}

\renewcommand{\thesection}{S\arabic{section}}
\renewcommand{\thesubsection}{S\arabic{section}.\arabic{subsection}}
\renewcommand{\thetable}{S\arabic{table}}
\renewcommand{\thefigure}{S\arabic{figure}}
\renewcommand{\theequation}{S\arabic{equation}}
\renewcommand{\theHsection}{S\arabic{section}}
\renewcommand{\theHsubsection}{S\arabic{section}.\arabic{subsection}}
\renewcommand{\theHtable}{supplement.table.\arabic{table}}
\renewcommand{\theHfigure}{supplement.figure.\arabic{figure}}
\renewcommand{\theHequation}{S\arabic{equation}}
\setcounter{section}{0}
\setcounter{table}{0}
\setcounter{figure}{0}
\setcounter{equation}{0}

This supplement has three sections. Section~S1 gives the simulation details: data-generating mechanisms, local-contamination calibration for the Section~3.1 behavior study, method settings, and scripts. Section~S2 documents the credit-card analysis: preprocessing, full-data comparison, timing environment, and scripts. Section~S3 explains the sparse-grid computation used to evaluate the LOO-LBFP score exactly. This computation avoids constructing the full grid and avoids refitting a leave-one-out density estimator for each observation.

\section{Simulation study details}\label{app:simulation_settings}

This section describes only the simulation components used in Sections~3.1--3.3 of the manuscript. The simulated outlier labels are used only to audit the contamination regimes and to compute external performance summaries.

\subsection{Main simulation designs}

The simulation study has two blocks. The first block is used in Section~3.1 to study the behavior of \(D_i\). It varies the sample size and the target local contamination count. It uses \(n\in\{1000,10000,100000\}\), contamination fraction \(\varepsilon=0.05\), \(R=20\) final replications, and four target local contamination counts \((\lambda)\). The pilot calibration for this block searches over \(\tau\in[0.03,300]\) on a logarithmic grid. The two inlier models are a standard bivariate normal distribution and a four-state Markov-switching Gaussian mixture adapted from \citet{CarbonDuchesne2024LBFP}.

The second block is used in Sections~3.2--3.3 to compare methods. It uses four scenarios, fixes \(n=1000\), sets the contamination fraction to \(\varepsilon=0.05\), and uses \(R=100\) replications.

For the method-comparison scenarios, the simulation script sets a global seed before the replication loop. Stochastic methods use deterministic method-specific seeds of the form \(10000+100(r-1)+j\), where \(r\) is the replication number and \(j\) is the method index. No scenario-specific parameter tuning is performed.

\begin{table}[htbp]
\centering
\caption{Data-generating mechanisms for the four simulation scenarios reported in the manuscript.}
\label{tab:dgp_settings}
\small
\begin{tabularx}{\textwidth}{lXX}
\toprule
Scenario & Inlier distribution & Contaminating distribution \\
\midrule
Mixture & Three Gaussian clusters with centers \((-2,-2)\), \((2,2)\), and \((-2,2)\); proportions \(0.4\), \(0.3\), and \(0.3\); covariance \(0.3I_2\). & Uniform contamination on \([-6,6]^2\). \\
\addlinespace
Ring & \(\theta\sim U(0,2\pi)\), \(r\sim N(4,0.4^2)\), with observations \((r\cos\theta,r\sin\theta)\). & Gaussian contamination \(N_2(0,0.2I_2)\). \\
\addlinespace
Heavy-Tail & A 70/30 mixture of \(N_2(0,I_2)\) and independent Student-\(t_3\) coordinates. & Radial extremes with angle \(U(0,2\pi)\) and radius \(U(6,9)\). \\
\addlinespace
Markov-Switching & Four-state Markov-switching Gaussian mixture adapted from the dependent inlier design of \citet{CarbonDuchesne2024LBFP}. & Gaussian contamination \(N_2\{(6,6)^\top,10^2I_2\}\). \\
\bottomrule
\end{tabularx}
\end{table}

\subsection{\texorpdfstring{Section~3.1 behavior study: local-contamination calibration and results}{Section 3.1 behavior study: local-contamination calibration and results}}

For the Section~3.1 behavior study, outliers are generated as
\[
Z\sim N_2\{(6,6)^\top,\tau^2I_2\},
\]
where \(\tau\) is chosen by pilot calibration. For each candidate \(\tau\), a pilot contaminated sample is generated and the LBFP bin width and grid origin are computed from that sample. Let \(C_i=1\) denote a simulated contaminating observation, and let \(A_b(x)\) be the union of the \(2^d\) counting cells whose bin counts can enter the LBFP interpolation at \(x\). The empirical local contamination count for an outlier \(X_i\) is
\[
\widehat\lambda_i
=
\sum_{r:C_r=1}
\mathbf 1\{X_r\in A_b(X_i)\}
-1,
\]
where the subtraction removes the self-count. The pilot value associated with \(\tau\) is the average
\[
\overline{\lambda}_{\mathrm{pilot}}(\tau)
=
\frac{1}{\sum_i C_i}
\sum_{i:C_i=1}\widehat\lambda_i,
\]
averaged over the two pilot replications. For a target value \(\lambda_0\), the selected value of \(\tau\) minimizes
\[
\left|\log\{1+\overline{\lambda}_{\mathrm{pilot}}(\tau)\}
-\log(1+\lambda_0)\right|
\]
over the logarithmic candidate grid. The empirical \(\widehat\lambda\) reported in Table~\ref{tab:lambda_regimes_supp} is recomputed on the final simulation replications using the selected \(\tau\).

\begin{table}[H]
    \centering
    \caption{Detailed results for the Section~3.1 behavior study. Emp. \(\widehat\lambda\) is the realized average number of other contaminating observations in the LBFP contributing neighborhood of a \(C_i=1\) observation. The two \(D_i\) columns give the average raw scores within inliers and contaminating observations. ROC-AUC ranks \(C_i=1\) observations above \(C_i=0\) observations. Entries report means across replications, with standard deviations in parentheses where applicable.}
    \label{tab:lambda_regimes_supp}
    \input{tables/table_31_lambda_regimes.tex}
\end{table}

\subsection{Simulation method configurations}

The method configurations used in the simulation scripts are:
\begin{itemize}
    \item LOO-LBFP: the LBFP density estimates were computed with the \texttt{GLBFP} package \citep{GLBFPpackage}. We used the single-grid LBFP estimator. The bin width \(b\) was selected as in \citet{CarbonDuchesne2024LBFP}, using the bandwidth selection implemented in \texttt{GLBFP}.
    \item LOF: \texttt{dbscan::lof} with \(k=20\), following \citet{Breunig2000}.
    \item \(k\)NN: \texttt{dbscan::kNN} with \(k=20\), following \citet{Ramaswamy2000}.
    \item KDEOS: \texttt{DDoutlier::KDEOS} logic with \(k_{\min}=10\) and \(k_{\max}=20\), following \citet{SchubertZimekKriegel2014}.
    \item LOO-KDE: \texttt{ks::Hpi} and \texttt{ks::kde}, with normal-reference fallback when the plug-in selector fails \citep{Silverman1986,Scott1992}.
    \item iForest: the \texttt{isotree} implementation, using 100 trees and subsample size \(\min(256,n)\), following \citet{liu2008isolationforest}.
\end{itemize}

\subsection{Simulation reproducibility scripts}

In the replication package, the simulation artifacts can be regenerated with the following scripts:
\begin{enumerate}
    \item \texttt{scripts/sim\_31\_lambda\_regimes.R} regenerates the Section~3.1 behavior-study table and figure, including the \(\lambda\)-calibration audit.
    \item \texttt{scripts/run\_simulations.R} regenerates the Section~3.2 ROC-AUC, score-summary, and runtime tables.
    \item {\footnotesize\texttt{scripts/generate\_scenario\_overview\_ggplot.R}} regenerates the simulation scenario overview.
    \item \texttt{scripts/regenerate\_property\_figures.R} regenerates the Section~3.1 diagnostic figure.
    \item \texttt{scripts/generate\_computation\_figure.R} regenerates the runtime-by-sample-size figure.
\end{enumerate}

\section{Credit-card data-analysis details}\label{app:creditcard_reproducibility}

This section documents only the real-data credit-card analysis. Fraud labels are used only after scoring, for external evaluation.

\subsection{Credit-card analysis settings}

The analysis uses the real Credit Card Fraud Detection data \citep{dalpozzolo2015calibrating}. The data are read either from the local copy in \texttt{data/creditcard\_42175.rds} or from OpenML dataset 42175 \citep{Vanschoren2013OpenML,creditcard2013}. No synthetic fallback is used. The full-data comparison uses standardized versions of the 28 anonymized PCA variables V1--V28 and Amount. Time is not used. Labels are used only after scoring for external evaluation.

Direct LOO-KDE is not included in the full-data table in the manuscript because its direct implementation had not finished after 30 minutes.

The corresponding machine-readable files in the replication package are:
\begin{itemize}
    \item \texttt{results/creditcard\_full\_method\_comparison.csv};
    \item \texttt{results/creditcard\_full\_method\_comparison.rds};
    \item \texttt{results/creditcard\_full\_method\_scores.rds}.
\end{itemize}

\subsection{Computing environment for data analysis}

The timing notes record a MacBook Pro with Apple M4 Pro processor, 14 CPU cores, and 24 GB RAM, running macOS 26.5. The R environment used for the final rebuild was R version 4.5.0 (2025-04-11) on platform \texttt{aarch64-apple-darwin20}. Main package versions included \texttt{dplyr} 1.2.1, \texttt{tidyr} 1.3.2, \texttt{ggplot2} 4.0.2, \texttt{dbscan} 1.2.2, \texttt{ks} 1.15.1, \texttt{isotree} 0.6.1.4, and \texttt{Rcpp} 1.1.1. The full-sample iForest run requested all available threads, but the local \texttt{isotree} build reported that it was compiled without OpenMP support; its recorded time should therefore be read as the time for that local build.

\subsection{Credit-card reproducibility scripts}

In the replication package, the full-data credit-card comparison and formatted manuscript table can be regenerated with:
\begin{enumerate}
    \item {\footnotesize\texttt{scripts/creditcard\_full\_method\_comparison.R}}\par
    Runs the full-sample V1--V28 plus Amount credit-card comparison, excluding direct LOO-KDE.
    \item \texttt{scripts/format\_creditcard\_presentation.R} formats the main credit-card table from the generated CSV file.
\end{enumerate}

The manuscript and supplementary material can then be compiled from the \texttt{article/} directory with \texttt{latexmk}.

\section{Sparse-grid computational strategy}
\label{app:sparse_strategy}

This section describes the computational strategy used to evaluate the LOO-LBFP score. The goal is to compute the same statistic as in the manuscript, but without constructing the full grid and without recomputing a leave-one-out density estimator for every observation.

Let \(G\subset \mathbb{Z}^d\) denote the set of occupied LBFP counting cells. Each observation \(X_i\) is first assigned to its grid cell, and the cell counts are stored in a sparse map:
\[
N_k=\sum_{i=1}^n \mathbf{1}\{X_i \in C_k\},
\qquad k\in G.
\]
Cells not appearing in \(G\) are not stored and are interpreted as having count zero. This representation avoids allocating the full Cartesian grid, which may be very large in moderate dimension.

For a point \(x\), the LBFP estimator is a multilinear interpolation of nearby cell counts. A direct implementation could enumerate all interpolation vertices and query their counts. The sparse implementation instead searches only through prefixes that are compatible with occupied cells. The occupied cell indices are organized as a sparse prefix tree. At each dimension, the algorithm keeps only the candidate cell coordinates that can contribute to the interpolation at \(x\). If no occupied cell matches a current prefix, the branch is discarded. Empty regions of the grid are therefore skipped before the algorithm reaches the final interpolation vertices.

The procedure for computing all scores is as follows.

\begin{enumerate}
    \item Build the sparse counting map by assigning each observation to its LBFP counting cell.
    \item Store the occupied cell indices in a prefix structure, optionally ordering dimensions by the number of occupied coordinates to improve pruning.
    \item For each observation \(X_i\), evaluate \(\widehat f_b(X_i)\) by traversing only occupied prefixes compatible with the interpolation stencil of \(X_i\).
    \item Compute the self-weight \(W_i\), which is the contribution of \(X_i\) to \(\widehat f_b(X_i)\) under the fixed grid and bandwidth.
    \item Apply the closed-form leave-one-out identity
    \[
    D_i
    =
    \frac{1}{n-1}
    \left\{
    \frac{W_i}{V_b\widehat f_b(X_i)}
    -
    1
    \right\}.
    \]
\end{enumerate}

This algorithm is exact for the fixed-grid LBFP estimator used in the manuscript. It changes only the computational evaluation of the estimator. The statistical definition of \(D_i\), the bandwidth, the grid, and the leave-one-out formula are unchanged. The gain comes from two sources: Proposition 1 removes the need for \(n\) separate leave-one-out fits, and the sparse-grid representation avoids storing and scanning empty grid cells.

\FloatBarrier
\bibliography{article2_references}